\documentclass[jmp]{revtex4}
\usepackage{amsmath,amsfonts,amssymb,amsthm}
\usepackage[active]{srcltx}
\theoremstyle{plain}
\newtheorem{thm}{Theorem}[section]
\newtheorem{lemma}[thm]{Lemma}

\theoremstyle{remark}
\newtheorem{remark}{Remark}

\newtheorem{assump}{Assumption}
\newcommand{\be}{\begin{equation}}
\newcommand{\ee}{\end{equation}}

\newcommand{\dist}{{\, \rm dist}}
\newcommand{\norm}[1]{ \left \| #1 \right \|}

\def\sgn{{\rm{sgn}}}

\newcommand{\R}{\mathbb R}
\newcommand{\N}{\mathbb N}
\newcommand{\C}{\mathbb C}

\newcommand{\di}{\mathrm d}

\newcommand{\ep}{\qquad {\vrule height 7pt width 7pt depth 0pt}}

\begin{document}
\title[Hamiltonian--based quantum computation of low--rank
 matrices] {On the efficiency of Hamiltonian--based quantum
computation for low--rank matrices}

\author{Zhenwei Cao}
\author{Alexander Elgart}
\email{zhenwei@vt.edu}\email{aelgart@vt.edu}
\affiliation{Department
of Mathematics, Virginia Tech., Blacksburg, VA, 24061 }
\date{Revision date: 20 Jan 2012}
\begin{abstract}
We present an extension of Adiabatic Quantum Computing (AQC)
algorithm for the unstructured search to the case when  the  number
of marked items is unknown. The algorithm maintains the optimal
Grover speedup and includes a small  counting subroutine.

Our other results include a lower bound on the amount of time needed to perform a general Hamiltonian-based quantum search, a lower bound on the evolution time needed to perform a search that is valid in the presence of control error and a generic upper bound on the minimum eigenvalue gap for evolutions.

In  particular, we demonstrate that quantum speedup for the unstructured search
using AQC type algorithms may only be achieved under  very rigid control precision requirements.
\end{abstract}
%

\maketitle
%
%
%
%
%
%
%
\section{Introduction and main results}\label{sec:intro}
\subsection{Introduction}
Quantum computing is believed to possess more computational power
than classical computing on certain computational tasks.
For example,  factorization of large numbers can be feasible
once a quantum computer is built \cite{SHOR}. The basic paradigm which is usually used in the theoretical works on quantum computing is the so called  quantum circuit model (QCM), see e.g. \cite{Nielsen}, although the  practical realization of it is yet to be found. Farhi and his collaborators \cite{FGGS} had proposed the adiabatic quantum computing (AQC) as an alternative, constructive model for implementing a quantum computer. It was later realized that from a computational complexity point of view AQC is equivalent to all other models for universal quantum computation \cite{ADKS}.
\par
Grover's algorithm \cite{G}, originally derived in the framework of QCM, is one of the milestones in quantum computing. The problem it solves can be formulated as  following:
Given $F: \{0,1\}^n\rightarrow\{0,1\}$ and  knowledge that there exists a  unique element $x$ such that
$F(x)=1$, find $x$. It is clear that classically, one needs to check $F(x)$
for all $N=2^n$ values of $x$ to find the solution, so the time complexity
of doing so is $O(N)$. Grover's algorithm uses only $O(\sqrt{N})$ steps
to achieve the result, in the framework of QCM. This bound is
indeed proven to be optimal for query type of quantum algorithms, see
\cite{BBBV}. In the case where
there are $m$ (not necessarily $1$) values of $x$ for which $F(x)=1$,
a modified version of Grover's algorithm works in QCM if the number $m$
is known \cite{BBHT}. The technique for finding $m$  is called quantum counting and was developed for  QCM  in \cite{BHT}.
\par
The original motivation of introducing AQC was to  derive  the physically attainable  algorithm for solving optimization problems such as as satisfiability of Boolean formulas by encoding a cost function into the Hamiltonian.  One of the (very few) models for which AQC had been shown to produce a speedup is  Grover's search problem, addressed first in \cite{FG}. The further works in this direction (e.g. \cite{T, RC}) considered the original problem of Grover (that is the case $m=1$). The natural question that arises in this context is whether the Grover type running time, {\it e.g.} $O(\sqrt N)$, is still  optimal  for a more general class of problem Hamiltonian, characterized by $N\gg m>1$. That means that one wants to derive the analytic lower bound on the runtime of the algorithm,  as well as to construct the suitable realization of the algorithm  for which the running time (being  the upper bound on the optimal time) is comparable with this lower bound.   The partial result in this direction, namely the lower bound for some class of such models,  was established  in \cite{IM} (we postpone the more detailed discussion till the next section).

In this paper, we derive the analytic upper and lower bounds on the amount of time needed to perform a (Hamiltonian based but not necessary adiabatic) unstructured search for the case $N\gg m>1$. We also present a lower bound on the evolution time needed to perform a search that is valid in the presence of control error and a generic upper bound on the minimum eigenvalue gap for the family of the interpolating Hamiltonians used in AQC. In particular, we  show that in general the $O(\ln N/\sqrt{N})$ control precision is  necessary in order to achieve quantum speedup over classical computation for the small values of $m$.

\subsection{Bounds on the running time in Hamiltonian--based quantum
computation}
In the abstract setting of AQC, we are interested in finding the ground state
of the given problem Hamiltonian $H_F$, in the shortest possible time.
To this end, we consider a pair of hermitian $N\times N$ matrices $H_{I,F}$,
and will assume that $N\gg1$. Let $H(s)$ be the interpolating Hamiltonian
\begin{equation}\label{eq:H}
H(s):=(1-f(s))H_I+f(s)H_F\,,
\end{equation}
where $f$ is a monotone function on $[0,1]$ satisfying $f(0)=0$,
$f(1)=1$. The idea of AQC is to prepare the initial state of the system
$\psi(0)$ in a ground state $\psi_I$ of the
Hamiltonian $H_I$, and let the system evolve according to the (scaled)
Schr\"odinger equation:
\begin{equation}\label{eq:Sch}
i\dot\psi_\tau(s) \ = \ \tau H(s) \psi_\tau(s)\,,\quad \psi_\tau(0)\
= \ \psi_I\,.
\end{equation}
The adiabatic theorem (AT) of quantum mechanics ensures that under
certain conditions (see theorem \ref{thm:adi} below for details) the
evolution $\psi_\tau(1)$ of the initial state stays close to a
ground state of the problem Hamiltonian $H_F$. For AQC to be
efficient, the running ({\it i.e.} physical) time $\tau$ in
\eqref{eq:Sch} must be much smaller than $N$. One then can ask what
choice of the initial Hamiltonian $H_I$ and the parametrization
$f(s)$ minimizes $\tau$, and what the optimal value of $\tau$ is.

One of the parameters that enters into the upper bound for $\tau$ in
the standard AT is the minimal value $g$ of the spectral gap $g(s)$
between the ground state energy of $H(s)$ and the rest of its
spectrum. Consequently, the traditional approach \cite{FGGS} to AQC involves the
estimation of $g$. Excluding a very short list of interesting
situations for which $g$ can be explicitly evaluated (compilation of such examples can be found in
\cite{DMV}), it appears to be a hard problem. In some instances one
can get an idea of what size of $g$ could be by using the
first-order perturbation theory \cite{T}. In subsection \ref{subsec: Gaps} we present rigorous bounds on the size of the gap for the problem at hand, albeit  we don't use them explicitly in our study of AQC.

The main purpose of this work is to obtain the rigorous upper and lower
bounds on the optimal running time $\tau$ for a particular class of
problem Hamiltonians, satisfying
\begin{assump}\label{assump'}
The problem Hamiltonian is of the small rank: \[Rank(H_F):=m\ \ll
N\,.\]
\end{assump}
This hypothesis is fulfilled in particular for the generalized
unstructured search (GUS) problem,  see {\it e.g.} \cite{BHT}. Since we are interested in the dynamical evolution of the initial state for which shifting the energy results in the overall dynamical phase factor, the above assumption is equivalent to the following condition: Let $V$ denote the largest eigenspace of $H_F$. Then we require that $N-\dim V\ll N$.

It turns out that for such $H_F$ one can circumvent the standard AT, avoiding
the direct estimation of $g$. We will also see that the (nearly)
optimal parametrization $f(s)$ is in fact non adiabatic.

To formulate our results, we need to introduce some notation first:
Let $\{E_n^i\}_{n=1}^N$ ($\{E_n^f\}_{n=1}^N$) be a set of distinct
eigenvalues of $H_I$ (respectively $H_F$), enumerated in the
ascending order. It is allowed to the corresponding eigenvalues to
be degenerate. In what follows, we will denote by $P_I$ ($P_F$) the
eigenprojection of $H_I$ ($H_F$) onto $E_I:=E_1^i$ ($E_F:=E_1^f$),
and by $Q_I$ ($Q_F$) the orthogonal projection onto the range of
$H_I$ ($H_F$). To AQC to be meanigful in our context we have to impose $E_I\neq0$. In the typical setup, $E_I=-1$.

Before stating our results, let us note that for AQC to work, it
suffices to ensure that $\psi_\tau(1)$ has just the non trivial
overlap with the range of $P_F$, which we will encode in the
requirement $\|P_F\psi_\tau(1)\|\ge\gamma$ for a "reasonable"
$\gamma$. Indeed, like many quantum algorithms, the AQC algorithm is
probabilistic in the sense that it gives the correct answer with the
probability $\gamma^2$. The probability of failure can be decreased
to the desired value (namely $O(1/N)$) by repeating the algorithm
$\frac{\ln N}{\gamma^2}$ times. We set $\gamma=1/5$ throughout this
paper. Another issue that we want to settle is normalization of
$H(s)$. To that end, we will calibrate $H_{I,F}$ as
$\|H_I\|=\|H_F\|=1$. Note that without loss of generality we can
assume that $E_F<0$ (since otherwise we can interpolate $-H_I$ and
$-H_F$ which  only changes the solution $\psi_\tau$ of
\eqref{eq:Sch} into $\bar \psi_\tau$). We now introduce some
parameters in order to formulate our results. Namely, let $\delta_1
=\| H_F\psi_I\|$, let $\delta_2 =\| P_F\psi_I\|$, and let $\delta_3
=\| Q_F\psi_I\|$, where $Q_F$ is a projection onto $Range\, H_F$.
Let $g_F:=E_2^f-E_1^f$.

Finally we introduce the notion of what we will refer to as a generic Hamiltonian $H_I$. Given an $m$--dimensional subspace $V$ of $C^N$, the natural question one can ask is what is a distance from the "typical" vector $\psi_I$ to $V$. More specifically, suppose one has some reasonable probability distribution function  for the vectors $\psi_I$ on the unit sphere $S^N$ in $C^N$ (say uniform). Then the expected value of $\|\phi_I\|^2$ of the orthogonal projection $\phi_I$ of the $\psi_I$ on $V$ is equal to $m/N$. One can  check that the probability of the event  $\{\psi_I\in \C^N:\ |\|\phi_I\|^2-m/N|\ge \alpha m/N\}$ is exponentially small in  $\alpha$ (see {\it e.g.} \cite{DG}). Note now that $Q_F\psi_I$ is the projection of $\psi_I$ onto the range of $Q_F$, which is an $m$--dimensional subspace.  We therefore will call  $H_I$ {\it generic} if its ground state $\psi_I$ satisfies $\|Q_F\psi_I\|=O(\sqrt{m/N})$.

Our first assertion is the non-existence result, showing that for any choice of $H_I$ and any function $f(s)$ the running time cannot be smaller than $\tau_-$ defined below.
\begin{thm}[The lower bound on the running time]\label{thm:at1}
Consider the interpolating family Eq.~\eqref{eq:H} with an arbitrary
 $f$. Then the running time $\tau_-$ in Eq.~\eqref{eq:Sch} for
which $\|P_F\psi_{\tau_-}(1)\|\ge1/5$ satisfies
\begin{equation}\label{eq:mintime}
\tau_-\ \ge \ \frac{1-5\| P_F\psi_I\|}{5\| H_F\psi_I\|}\,,\quad \mbox{ for
}\quad\delta_2<1/5\,.
\end{equation}
\end{thm}
\begin{remark}
\begin{enumerate}
\item[]
\item This result shows that it is impossible to construct the family of the interpolating Hamiltonians $H(s)$ such that the evolution of $\Psi_I$ will have a reasonable overlap with $P_F$ if the running time $\tau$ is smaller than $\tau_-$. To probe how tight this bound is, one wants to construct a specific family $H(s)$ and the running time $\tau_+$ for which $\|P_F \psi_\tau(1)\|$ is not small, and make a comparison between $\tau_\pm$. We construct such $H(s)$ in the next assertion.   As we shall see, our bound $\tau_-$ is not tight ($\tau_-/\tau_+ -1\neq o(1)$), but of the right order of magnitude (meaning $\tau_-/\tau_+= O(1)$) in terms of the asymptotic dependence on the small parameter $m/N$.
\item For a generic $H_I$ both $\| H_F\psi_I\|$ and $\| P_F\psi_I\|$ are
$O(\sqrt{m/N})$, hence the minimal running time $\tau$ cannot be
smaller than $O(\sqrt{N/m})$.
\item As we will see, the (nearly) optimizing parametrization $f(s)$
is in fact non adiabatic.
\end{enumerate}

\end{remark}

\vspace{.5cm}

\noindent {\bf Comparison with the Ioannou - Mosca result}.

\vspace{.5cm}

In \cite{IM}, Ioannou and Mosca established the lower bound on the running time $\tau$ for a particular class of problems where the initial Hamiltonian $H_I$ is diagonal in the Hadamard basis while the problem Hamiltonian $H_F$ is diagonal in the standard basis. Their result is non trivial provided the largest eigenspace  of $H_I$ has dimension $N-m$ where $m\ll \sqrt N$, and the lower bound they obtained is given by $\tau_-=O(\sqrt N/m)$. Since one can always interchange the roles of $H_I$ and $H_F$ and shift energy so that the largest eigenspace corresponds to the energy $0$, their result  can be viewed as a slighter
weaker version of Theorem \ref{thm:at1} for this class of Hamiltonians.

\vspace{.5cm}
\noindent
In the next assertion we construct a specific family $H(s)$ and determine the runtime $\tau_+$ for which $\|P_F \psi_\tau(1)\|\ge 1/5$:
\begin{thm}[The upper bound on the running time]\label{thm:at2}
There exists an explicit
rank one  $H_I$ and an explicit function $f$ such that
$\|P_F\psi_{\tau_+}(1)\|\ \ge \ 1/5$ for
\begin{equation}\label{eq:mintime'}
\tau_+\ = \ \frac{C(1-E_F)}{|E_F|\,\| P_F\psi_I\|}\,,
\end{equation}
and any value $C\in[1/3,2/3]$, provided  that $\| Q_F\psi_I\|/g_F=O(1/\ln N)$.
\end{thm}
\begin{remark}
\begin{enumerate}
\item[]
\item
For a generic choice of $H_I$ one has $\| Q_F\psi_I\|=O(\sqrt {m/N})$,
$\| P_F\psi_I\|=O(\sqrt {m_1/N})$, where $m_1=Rank\,P_F $.  It
implies that $\tau_-/\tau_+=O(1)$  for $m=O(1)$.
\item
This assertion can be viewed as an extension of the
result obtained in \cite{FG} that considered the original Grover's search problem in the Hamiltonian--based algorithm.
\item The interpolating function $f$ in this construction is similar to the one used in \cite{FG}, namely it is a double step function. Since $f$ is discontinuous, we prefer to refer to this particular construction as  the  Hamiltonian--based algorithm rather than AQC.
\end{enumerate}
\end{remark}
Note that a-priori the values of $E_F$ and $\delta_2$ may be
unknown.  For instance, the value of the overlap $\delta_2$ has to
be determined in the GUS problem with the {\it unknown} number of
marked items. To this end, we prove the following auxiliary result:
\begin{thm}\label{thm:delta}
Suppose that the value of $E_F$ is known. Then there is a
Hamiltonian -- based algorithm that determines $\| P_F\psi_I\|$ with $1/
N^2$ accuracy and requires $O((\ln N)^2)$ of the running time.
\end{thm}
\begin{remark}
\begin{enumerate}
\item[]
\item
The running time for this sub-algorithm is much shorter than
$\tau_+$, so it does not significantly affect the total running
time.
\item
A parallel result in the context of the quantum circuit model
was establishes earlier in \cite{BHT}.
\end{enumerate}
\end{remark}
\subsection{Gaps in the spectrum of the interpolating
Hamiltonian}\label{subsec: Gaps} Although the size of the gap in the
spectrum of $H(s)$ did not play much of the role so far, it is
instructive to estimate it for the following reason: The size of the
gap manifests itself in the adiabatic theorem of quantum mechanics
(AT), on which AQC is build.
The following assertion holds, see {\it e.g.} \cite{AE}:
\begin{thm}[Uniform adiabatic theorem]\label{thm:adi}
Suppose that the  $H(s)$ is twice differentiable and bounded family of self adjoint operators on the interval $[0,1]$ that is $\tau$--independent, and suppose in addition that
\be\label{eq:gap_g}
g\ := \ \dist(\lambda_1(s),\sigma(H(s)\setminus\lambda_1(s)) \ > \ 0 \quad \mbox{  for all } s\in[0,1]\,.
\ee
Then the solution
$\psi_\tau(s)$ of the IVP \eqref{eq:Sch} satisfies
\begin{equation}\label{eq:adi}
\lim_{\tau\rightarrow\infty}\dist\left(\psi_\tau(s),Range\
P_F\right) \ = \ 0\,.
\end{equation}
\end{thm}
To AQC to be meaningful, one should choose the initial Hamiltonian
$H_I$ in such a way that $Rank\ P_I$ is small.  The error in the adiabatic
evolution (the right hand side of \eqref{eq:adi}) depends on the size of the gap $g$, with the rough upper
bound on the error of the form $\frac{C}{\tau g^3}$ \cite{JRS}.

\begin{thm}[The size of the gap]\label{thm:gap}
Let $g_I:=E_2^i-E_I$ be a gap between the ground state of the
initial Hamiltonian $H_I$ and the rest of its spectrum. Let
$\delta_4: =\|P_I Q_F\|$, where $Q_F$ is a projection onto the range
of $H_F$. Then we have the following estimate on the size of the gap
$g$ in \eqref{eq:gap_g}:
\begin{equation}\label{eq:gap}
 g \ \le \ 10\,\delta_4\,,
\end{equation}
provided $g_I>10\delta_4$.
\end{thm}
\begin{remark}
In fact one can relax the condition $g_I>10\delta_4$, but to keep the
presentation simple we impose this additional constraint.
\end{remark}
\subsection{Robust adiabatic quantum computing}
In this section we propose a necessary technical requirement on the
quantum device for AQC Grover's search to be successful.
The $\sqrt N$ speedup
in AQC algorithm obtained in the tractable problems ({\it c.f.}
\cite{DMV} for the Grover's problem or its rigorous treatment in
\cite{JRS}) relies on a special choice of the parametrization $f(s)$ in
\eqref{eq:H}. Namely, it is constructed in such a way that $\dot
f(s)$ is  small at instances $\{s_j\}$ at which the spectral gap
$g(s_j):=\lambda_2(s_j)-\lambda_1(s_j) $ of $H(s_j)$  is  small. In the AQC jargon, it is usually referred to as the quantum search by local adiabatic evolution. It is interesting to compare this approach with the construction used in Theorem \ref{thm:at2}, where this strategy is pushed to the extreme, namely $f$ used there is actually the constant except for the endpoints $s=0,1$ where it jumps.
There are two practical problems with this approach:    
\begin{enumerate}
\item The values  $\{s_j\}$ obviously depend on $H_F$ and in particular
on $E_F$ (even for the Grover's problem, as the simple scaling
argument shows). So to choose such an $f$ one has to know the
spectral structure of $H_F$ with $o(1/\sqrt N)$ precision. This is
tacitly assumed in \cite{DMV}. Note that albeit Theorem
\ref{thm:at2} (used in conjunction with Theorem \ref{thm:gap})
represents an improvement with this regard, it still requires
knowledge of $E_F$.
\item Even if this technical obstacle can be overcome, the extreme susceptibility
of $\psi_\tau(1)$  to the parametrization $f$ poses a radical
problem in practical implementation.  Indeed, it is presumably
extremely difficult to enforce $\dot f=0$ for a long stretch of the
physical time, as the realistic computing device inevitably
fluctuates due to the presence of the noise. Some models that try to take into the account the noise were proposed, see {\it e.g. } \cite{CF,AAN}, but to the best of our knowledge all of the existing constructions contain ad hoc parameters and are not derived from the first principles. For some interesting rigorous work in this direction that considers de-phasing open systems see \cite{avron}.
\end{enumerate}
Another issue \footnote{We thank the referee for bringing this point to our attention.} that will motivate our last result below is related to the fact that the adiabatic theorems fall into two categories: Those that describe the solutions for {\it all} times, including times $s \in [0, 1]$, and those that characterize the solutions at large times $s > 1$ where the Hamiltonian is time independent again. Interestingly they give more precision for long times. We call the first category, the one that applies to all times, uniform, the second is the long time category.

A representative result from the uniform category is Theorem \ref{thm:adi} above. A characteristic result (see, {\it e.g.} \cite{berry, nenciu,martinez}) which lies in the long time category is
\begin{thm}[Long time adiabatic theorem]\label{thm:adi_long}
Suppose that the  $H(s)$ is smooth (that is $C^\infty$ class) and bounded family of self adjoint operators  with $\dot H (s)$ supported on $[0,1]$ that is $\tau$--independent, and suppose in addition that \eqref{eq:gap_g} holds as well. Then the solution
$\psi_\tau(s)$ of the IVP \eqref{eq:Sch} satisfies
\begin{equation}\label{eq:adi_long}
\dist\left(\psi_\tau(s),Range\
P_F\right) \ = \ o(\tau^{-n})\mbox{ for } s\ge1\,,
\end{equation}
for any $n\in\N$.
\end{thm}
\begin{remark}
\begin{enumerate}
\item[]
\item In words one can say that starting and finishing the interpolation slowly decreases the error in the adiabatic theorem.
\item There is, in general, no uniformity in $n$; the term on the right hand side is of order $c_n \tau^{-n}$ where $c_n$ grows rapidly with $n$ ({\it c.f.} the following  discussion).
\item The distinction between the uniform and the long time AT has an analog in integrals. Suppose that $g(s)\in C^{\infty}([0, 1])$. Then
\[\int_0^s g(t) e^{it\tau}dt \ = \ \begin{cases}
&o(\tau^{-n})\,,\mbox{ if } s\ge 1\, ; \\
 &   O(\tau^{-1}) \mbox{ if }s\in(0,1)\,.
\end{cases}\]
\end{enumerate}
\end{remark}
In the application to AQC it is natural  to investigate the dependence of the coefficients $c_n$ in terms of the gap $g$ and minimize the running time $\tau$ in such a way that $c_n \tau^{-n}=o(1)$ for some optimally chosen value $n$.  The recent result in this direction, \cite{LRH}, gives $\tau=O( g^{-3})$. For the sketch of the argument that uses (truncated) Nenciu's expansion technique \cite{ESn} and leads to the sharper estimate  $\tau=O( g^{-2}|\ln g|^5)$ see \cite{elgart}.  One is then tempted to combine the starting and finishing  slowly strategy with the quantum search by local adiabatic evolution strategy in order to minimize the error in the adiabatic theorem. Such analysis was undertaken recently for Grover's search problem in \cite{RPL,WB}.

Inspired by the above discussion, we will assume that in the robust setting  for any given moment $s$ inside the interval $J$ described below (and which excludes the vicinities of the endpoints $s=0,1$) the value $\dot f(s)$ is greater than some small but fixed $\kappa>0$. To motivate the definition of  $J$, suppose that the  function $f$ lies in the long time category, {\it i.e.}  $\dot f$ is supported in $[0,1]$ and $f$ is smooth. Let $b\in[0,1]$ be such that $\ddot f$ does not change sign on $[b,1]$  (but it can vanish there). It is not difficult to see that since $f$ is monotone, $f$ has to be concave on  $[b,1]$, hence $\ddot f\le0$ there. Now let us define the interval $J$ for {\it any}  differentiable function $f$. Let $a=\min_{s\in[0,1]}\{f(s)=1/3\}$. Let $b=\min_{s\in[0,1]}\{f(s)\mbox{ is concave on }[s,1]\}$. We then define the interval $J:=[a,b]$ for $f\in C^1$ provided $a\le b$, $J=\emptyset$ if $a> b$, and $J:=[a,1]$ if $f\notin C^1$.  To illustrate this notion, consider  $f\in \C^\infty(\R)$  constructed as follows:
\[f(t) \ = \ \int_{-\infty}^t g(s) ds\,,\quad g(s) \ = \  \begin{cases}
&0\,,\mbox{ if } s\notin[0,1]\, ; \\
 &  \alpha \, e^{\frac{1}{s(s-1)}} \mbox{ if }s\in(0,1)\,.
\end{cases}\]
The factor $\alpha$ here is a normalization constant, chosen so that $f(1)=1$. We then have
\[\ddot f(t) \ = \ \left(\frac{1}{t^2}\,-\, \frac{1}{(1-t)^2}\right) g(t)\,,\mbox{ for } t\in[0,1]\,,\]
so that the only inflection point is $t=1/2$. Hence $f$ is convex on $[0,1/2]$ and is concave on $[1/2,1]$. We therefore get $b=1/2$ and $J=[f^{-1}(1/3),1/2]$. The convexity of $f$ on $[0,1/2]$ implies that $f(y)-f(x)\ge \dot f(x)(y-x)$ for any $x,y\in[0,1/2]$. Choosing $x=s$, $y=0$, we obtain $\dot f(s)\ge f(s)/s$ for $s\in(0,1/2]$. Since $f$ is monotone, we conclude that $\dot f(s)\ge 2/3$ on $[a,1/2]$. So in this example $\kappa=2/3$.

\par
The utility of the introduction of the interval $J$ is as follows:  On the interval $[a,1]$ the function $f$ is concave, hence it satisfies $f(y)-f(x)\le \dot f(x)(y-x)$ for any $x,y\in[b,1]$. In particular, we have
\be\label{eq:concave}
1-f(t) \ \le \ \dot f(t)(1-t)\ \le \ \dot f(t)\mbox{ for } t\in[b,1]\,,
\ee
the relation we are going to exploit.

Our last result establishes that in the case of the small rank initial Hamiltonian the robust version of AQC does not yield a significant speedup unless $\kappa$ can be made exponentially small:
\begin{thm}[Robust lower bound on the running time]\label{thm:at5}
Suppose that $f$ in Eq.~\eqref{eq:H} is (piecewise) differentiable and satisfies
$\dot f(s)\ge\kappa>0$ for $s\in J$ with the interval $J$ defined above. Also, let us assume that $E_I=-1$. Then, if
$\tau<\tau_r=O\left(\frac{\kappa}{m^2\delta^2\,\ln\delta}\right)$,
we have
\begin{equation}\label{eq:mintime''}
\left|\langle\psi_I|\psi_\tau(1)\rangle\right|\ >\
\frac{2\sqrt{6}}{5}\,+\,2\delta\,,
\end{equation}
where $\delta=\norm{Q_I Q_F}$. Hence the running time $\tau$ for
which $\|Q_F\psi_{\tau}(1)\|\ge1/5$ cannot be smaller than $\tau_r$.
\end{thm}
\begin{remark}
\begin{enumerate}
\item[]
\item
This theorem tells us that for a generic $H_I$ of the small rank the robust running time $\tau_r$ cannot be smaller than $O(\kappa N/\ln N)$. Hence  unless the control precision $\kappa$ is on the order of $O(\ln N/\sqrt{N})$, AQC is not much better than its  classical counterpart that solves GUS for $\tau=O(N)$.
\item As we remarked earlier, the requirement $E_I=-1$ is a very mild one.
\end{enumerate}
\end{remark}
\section{Proof of Theorem \ref{thm:at1}}
The proof is based on the following observation: Note that $\psi_I$ is an approximate eigenvector of $H(s)$ since $(H(s)-(1-f(s))E_I I)\psi_I=f(s)H_F\psi_I$, and the norm of the right hand side is equal to $\delta_1$. So by the first order perturbation theory, the dynamical evolution of the state $\psi_I$ given by \eqref{eq:Sch} will stay close (up to the dynamical phase) to $\psi_I$, unless the running time $\tau$ is such that the total variation, given by  $\tau \delta_1$,  is of order $1$. The proof below formalizes  this argument.

For a solution $\psi_\tau(s)$ of \eqref{eq:Sch}, let
\be\label{eq:phi'} \phi_\tau(s):= \
e^{if_\tau(s)}\psi_\tau(s)\,,\quad f_\tau(s)=\tau\,E_I\int_0^s
(1-f(r))\,\di r\,.\ee
Then one can readily check that $\phi_\tau(s)$ satisfies the IVP
\begin{equation}\label{eq:Schd'}
i\dot\phi_\tau(s) \ = \ \tau \hat H(s) \psi_\tau(s)\,,\quad
\phi_\tau(0)\ = \ \psi_I\,,
\end{equation}
where
\[\hat H(s)\ = \ (1-f(s))\,(H_I-E_I\,I)\,+\,f(s)\,H_F\,.\]
The factor $e^{if_\tau(s)}$ is usually referred to as a dynamical
phase.
\par
Let $U_\tau(t,s)$ be a semigroup generated by $\hat H(s)$, namely
\begin{equation}\label{eq:semigroup}
-i\,\partial_s \,U_\tau(t,s) \ = \ \tau\,U_\tau(t,s)\hat
H(s)\,;\quad
 U_\tau(s,s)\ = \ I\,; \quad t\ge s\,.
\end{equation}
Then the solution $\phi_\tau(1)$ of \eqref{eq:Schd'} is equal to
$U_\tau(1,0)\,\psi_I$. On the other hand,
\[I\,-\,U_\tau(1,0) \ = \ \int_0^1 \partial_s \,U_\tau(1,s)\,\di s\
= \ i\,\tau\, \int_0^1U_\tau(t,s)\hat H(s)\,\di s\,,\]
hence applying both sides on $\psi_I$ we obtain
\[\psi_I\,-\,\phi_\tau(1) \ = \ i\,\tau\, \int_0^1U_\tau(t,s)\hat H(s)\psi_I\,\di
s\,.\]
We infer
\[\|\psi_I\,-\,\phi_\tau(1)\| \ \le \tau\, \int_0^1\|\hat H(s)\psi_I\|\,\di
s\,.\]
But
\[\hat H(s)\psi_I \ = \
\left\{(1-f(s))\,(H_I-E_I\,I)\,+\,f(s)\,H_F\right\}\psi_I \ = \
f(s)\,H_F\psi_I\,,\]
and we get the bound
\[\|\psi_I\,-\,\phi_\tau(1)\| \ \le \ \tau\, \delta_1 \int_0^1 f(s)\,\di
s\ \le \ \tau\, \delta_1\,,\]
where in the last step we used $0\le f(s)\le1$. By the triangle inequality,
\begin{eqnarray*}\Big|\|P_F\psi_I\|\,-\,\|P_F\phi_\tau(1)\| \Big| & \le &
\|P_F\psi_I\,-\,P_F\phi_\tau(1)\| \ \le \
\|\psi_I\,-\,\phi_\tau(1)\|\\ &\le & \tau\, \delta_1\,,\end{eqnarray*}
so that
\[\|P_F\phi_\tau(1)\| \ \le \ \tau\, \delta_1 \,+\,\delta_2\,.\]
 On the other hand, by the assumption of the theorem
$\|P_F\psi_\tau(1)\|\ge1/5$, hence $\|P_F\phi_\tau(1)\|\ge1/5$. As a
result, we can bound
\[1/5 \ \le \ \tau\, \delta_1 \,+\,\delta_2\,,\]
and the assertion follows.
 \hfill\ep
\section{Proof of Theorem \ref{thm:at2}}
We will choose $H_I=-|\psi_I\rangle\langle\psi_I|$, and a non adiabatic parametrization
\[
f(s) \ = \ \begin{cases}
&0\,,\mbox{ if } s=0\, , \\
 &   \alpha\equiv\frac{1}{1-E_F} \mbox{ if }s\in(0,1)\,,\\
 &1\mbox{ if }s=1\,.
\end{cases}
\]
That means we move extremely quickly (instantly in fact) to the
middle of the path, stay there for the time $\tau$, and then move
quickly again to the end of the path. We first observe that
regardless of the choice of $f(s)$ in \eqref{eq:H} we have
$\psi_\tau(s)\in Y$, where $Y$ is a subspace of the Hilbert space,
spanned by the vectors in the range of $H_F$ and $\psi_I$. Here we
have used the fact that the range of $H_I$ by the assumption of the
theorem coincides with $Span\{\psi_I\}$. Let us choose the
orthonormal basis $\{e_i\}_{i=1}^{m+1}$ for $Y$ as follows: The
first $m$ vectors in the basis are the eigenvectors of $H_F$
corresponding to $\{E_i^f\}$ that differ from zero, and the last
vector $e_{m+1}$ is obtained from $\psi_I$ using the Gram Schmidt
procedure. That is,
\[e_{m+1}:=\ \frac{\bar Q_F\psi_I}{\|\bar Q_F\psi_I\|} \ = \
\frac{\bar Q_F\psi_I}{\sqrt{1-\delta_3^2}} \,,\quad \bar Q_F\ =\
1-Q_F\,.\]
We then have
\begin{eqnarray}\label{eq:dist}
\|e_{m+1}-\psi_I\|^2& =& \|Q_Fe_{m+1}-Q_F\psi_I\|^2\,+\, \|\bar
Q_Fe_{m+1}-\bar Q_F\psi_I\|^2\nonumber \\ &=& \delta_3^2\,+\,
\left(\frac{1}{\sqrt{1-\delta_3^2}}\,-\,1\right)^2
\left(1-\delta_3^2\right)\nonumber  \\ &\le & \delta_3^2
\,+\,\delta_3^4 \,.
\end{eqnarray}
Our choice of $g$ ensures that
\[\psi_\tau(1) \ = \ e^{-i\alpha\tau\cdot(E_FP_I\,+\,H_F)}\psi_I\,.\]
Here we introduce $P_F$ as the orthogonal projection onto the span
of $\{e_i\}$ with $E^f_i=E^f_1$, and $P_{m+1}$ the orthogonal
projection onto $e_{m+1}$. Clearly
\[ P_{m+1}\ = \ \frac{\bar Q_F\,P_I\,\bar Q_F}{1-\delta_3^2} \]
We want to compute the matrix elements of the propagator
$e^{-i\alpha\tau(H_I\,+\,H_F)}$ in the basis $\{e_i\}$. To this end,
we observe that in this basis $H_F=diag(E_F,\ldots,E_m^f,0)$, and
$H_I$ is a block matrix such that
\begin{equation}\label{eq:HIest}
\left\|E_FP_I \, -\, \left[\begin{array}{cc}
0 & \delta_3\,V^* \\
\delta_3\, V &E_F\\
\end{array}\right]\right\|\ \le \ 3\delta_3^2 |E_F|\,,
\end{equation}
where $\|V\|=|E_F|$.
Indeed, we have
\[ \norm{Q_F\ P_I\ Q_F} \ = \,\delta_3^2\,,\]
\[ \norm{P_{m+1}\ P_I\ P_{m+1}}  \ = \ 1-\delta_3^2\,,\]
\[ \norm{P_{m+1}\ P_I\ Q_F}\ = \ \delta_3\,(1-\delta_3^2)^{1/2}\,.\]

Together, we obtain that in this basis
\[\left\|E_FP_I\,+\,H_F \,- \, \left[\begin{array}{cc}
D & \delta_3\,V^* \\
\delta_3\, V &E_F\\
\end{array}\right]\right\|\ \le\ 3\delta_3^2 |E_F|\,,\]
where $D=diag(E_F,\ldots,E_m^f)$. A simple perturbative argument
({\it cf.} Duhamel formula \eqref{eq:Du1} below) shows that
\begin{equation}\label{eq:K}
\left\|e^{-i\alpha\tau\cdot(E_FP_I\,+\,H_F)} \,- \, e^{-i\alpha\tau
K}\right\|\ \le \ 3\delta_3^2 |E_F|\alpha\tau\,,
\end{equation}
with
\[K \ = \ \left[\begin{array}{cc}
D & \delta_3\,V^* \\
\delta_3\, V &E_F\\
\end{array}\right]\,.\]
To this end, we split $K$ into the diagonal and off diagonal parts:
\[K \ = \ K_1\,+\,K_2: = \ \left[\begin{array}{cc}
D & 0 \\
0 &E_F\\
\end{array}\right]\,+\,\left[\begin{array}{cc}
0 & \delta_3\,V^* \\
\delta_3\, V &0\\
\end{array}\right]\,.\]
Let
\[\Omega(s):=\ e^{i\alpha\tau sK_1}\,e^{-i\alpha\tau sK}\,;\quad
K_2(s):=\ e^{i\alpha\tau sK_1}\,K_2\,e^{-i\alpha\tau sK_1}\,,\]
then $\dot \Omega(s)=-i\alpha\tau K_2(s)\Omega(s)$, and
$e^{-i\alpha\tau K}$ is given by the following Duhamel formula:
\begin{eqnarray}\label{eq:Du1}
e^{-i\alpha\tau K} & = &
e^{-i\alpha\tau K_1}\left\{I\,-\,i\alpha\tau \int_0^1K_2(s)\Omega(s)\,\di
s\right\} \\ &=&
e^{-i\alpha\tau K_1}\left\{I\,-\,i\alpha\tau \int_0^1
K_2(s)\,\di s\right.\nonumber \\
&&\hspace{1.5cm}+\,\left.\left(-i\alpha\tau \right)^2\int_0^1
K_2(s)\,\di s\int_0^s K_2(r)\Omega(r)\,\di
r\right\}\,.\label{eq:Du2}
\end{eqnarray}
Observe now that
\[
\left[K_2(s)\right]_{1,m+1}\ = \ \left[K_2\right]_{1,m+1}\,,\]
since $e^{-i\frac{\tau}{2}sK_1}$ is diagonal with $(1,1)$ entry
equal to $(m+1,m+1)$ entry. In fact, $\left[K_2(s)\right]_{j,m+1}\ =
\ \left[K_2\right]_{j,m+1}$ for all $j$ such that $E^f_j=E^f_1$.
Therefore
\begin{equation}\label{eq:snd}
\left\|P_F\,\int_0^1 K_2(s)\,\di s\,P_{m+1}\right\| \ =\
\delta_2\,|E_F|\,,
\end{equation}
since
\[ P_F\,K_2\,P_{m+1} \ = \ \frac{|E_F|}{\sqrt{1-\delta_3^2}}\,P_F\,P_I\,P_{m+1}\,.\]
To estimate the second term in \eqref{eq:Du2}, we note first that
the following bound holds:
\begin{lemma}\label{lem:kato}
We have
\begin{equation}\label{eq:snd2}
\left\|\int_0^s P_{m+1}\,K_2(r)\,Q_F\,\bar P_F\,\Omega(r)\,\di
r\right\| \ \le\ \frac{8\delta_3}{\tau g_F} \,+\,\frac{2\delta_3^2}{
g_F}\,.
\end{equation}
\end{lemma}
This estimate is essentially a content of Lemma 3.3 in \cite{ESr}.
The idea is that, since the spectral supports of $K_1$ for $P_{m+1}$
and $Q_F\bar P_F$ are a distance $ g_F$ apart, the integral over $r$
has a highly oscillating phase of order $\tau g_F$. For
completeness, we prove this lemma below.

\par
Armed with this estimate and using the fact that $K_2(s)$ is off
diagonal, we get
\begin{eqnarray*}
&&\left\|P_F\,\int_0^1 K_2(s)\,\di s\int_0^s K_2(r)\Omega(r)\,\di r\,P_{m+1}\right\| \\&&
= \ \|P_F\,K_2\,P_{m+1}\|\cdot
\left\|P_{m+1}\,\int_0^1\,\di s\int_0^s K_2(r)\Omega(r)\,\di r\,P_{m+1}\right\| \\
&&\le \ |E_F|\,\delta_2\,\cdot \int_0^1\,\di s
\norm{\int_0^s\,P_{m+1}K_2(r)(P_F+Q_F\bar P_F)\Omega(r) P_{m+1}\di r}\\
&&\le \ |E_F|\,\delta_2\,\cdot
\left\{ \int_0^1 \left\|P_{m+1}\,K_2(r)\,P_F\right\|\,\di r\right. \\&&
\left.+\, \max_{s\in [0,1]} \left\|\int_0^s
P_{m+1}\,K_2(r)\,Q_F\,\bar P_F\,\Omega(r)\,\di r\right\|\right\}\,,
\end{eqnarray*}
where we have used $\|\Omega(r)\|=\|P_{m+1}\|=1$ and \eqref{eq:snd}.
Applying  estimates \eqref{eq:snd} and \eqref{eq:snd2}, we bound
\begin{eqnarray}
&&\left\|P_F\ \int_0^1  K_2(s)\,\di
s\int_0^s K_2(r)\Omega(r)\,\di r\ P_{m+1}\right\|
\nonumber \\ &&\hspace{1cm}  \le  \
|E_F|\,\delta_2\,\cdot
\left\{|E_F|\,\delta_2\,+\,\frac{8\delta_3}{\tau g_F}
\,+\,\frac{2\delta_3^2}{ g_F}\right\}\,.\label{eq:offbnd}
\end{eqnarray}
Multiplying \eqref{eq:Du2} by $P_F$ from the left and by $P_{m+1}$
from the right, and using the estimates \eqref{eq:snd} and
\eqref{eq:offbnd}, we establish
\begin{eqnarray}
&& \left\|P_F\,e^{-i\alpha\tau K}\,P_{m+1}\right\|  \ \ge \
\alpha\tau\,\left\|P_F\,\int_0^1 K_2(s)\,\di s\,P_{m+1}\right\| \nonumber \\
&& -\,(\alpha\tau)^2\,\left\|P_F\, \int_0^1 K_2(s)\,\di s\int_0^s
K_2(r)\Omega(r)\,\di r\,P_{m+1}\right\| \nonumber \\
&&\hspace{1cm}= \
\alpha\tau|E_F|\,\delta_2\cdot\left(1\,-\,\alpha\tau\,\cdot
\left\{|E_F|\delta_2\,+\,\frac{8\delta_3}{\tau g_F}
\,+\,\frac{2\delta_3^2}{ g_F}\right\}\right)\,.\label{eq:Ka}
\end{eqnarray}

Combining the estimates in \eqref{eq:dist}, \eqref{eq:K}, and
\eqref{eq:Ka}, the result will follow provided
\begin{eqnarray*} &&\alpha\tau|E_F|\,\delta_2\cdot\left(1\,-\,\alpha\tau\,\cdot
\left\{|E_F|\delta_2\,+\,\frac{8\delta_3}{\tau g_F}
\,+\,\frac{2\delta_3^2}{ g_F}\right\}\right)\\ &&\hspace{1cm}\ge
\,1/5\,+\,\delta_3^2\,+\,\delta_3^4\,+\,3\delta_3^2|E_F|\alpha\tau\,.
\end{eqnarray*}
Note now that for $\tau=O(1/\delta_2)$ the above inequality is
satisfied for values of $\tau$ and $\delta_3$ such that
\[ \alpha\tau|E_F|\,\delta_2\cdot\left(1\,-\,\alpha\tau |E_F|\delta_2\,\right)\,
\ge \,2/9\,,\quad \delta_3/ g_F=O(1/\ln N)\,.\] The result now
follows. \hfill\ep
\begin{proof}[Proof of Lemma \ref{lem:kato}]
Let
\begin{equation}\label{eq:X}
X:=\ \frac{1}{2\pi i}\oint_\Gamma P_{m+1}\,(K_1-zI)^{-1}\,
K_2\,(K_1-zI)^{-1}\,Q_F\,\bar P_F\,\di \,z\,,
\end{equation}
where the contour $\Gamma$ is a circle $\{z\in\C:\ |z-E_F|=
g_F/2\}$. Since
\[\frac{1}{2\pi i}\oint_\Gamma (K_1-zI)^{-1}\,\di \,z \ = \
P_F\,+\,P_{m+1}\,,\]
one can readily check that
\[[X\,,\,K_1]\ = \ P_{m+1}\,\,K_2\,Q_F\,\bar P_F\,.\]
Hence
\begin{eqnarray*}&& \int_0^s P_{m+1}\,K_2(r)\,Q_F\,\bar P_F\,\Omega(r)\,\di r \\ && \hspace{1cm} = \
\frac{-2i}{\tau}\, \int_0^s \frac{d}{\di r}\Big(
e^{-i\frac{\tau}{2}rK_1}\,X
\,e^{i\frac{\tau}{2}rK_1}\Big)\,\Omega(r)\,\di r\,.\end{eqnarray*}
Integrating the right hand side by parts, we obtain
\begin{eqnarray*}
&&\int_0^s P_{m+1}\,K_2(r)\,Q_F\,\bar P_F\,\Omega(r)\,\di r \\ &&= \
-\frac{2i}{\tau}\,\left\{ e^{-i\frac{\tau}{2}rK_1}\,X
\,e^{i\frac{\tau}{2}rK_1}\,\Omega(r)\Big|_0^s\,-\,\int_0^s
e^{-i\frac{\tau}{2}sK_1}\,X \,e^{i\frac{\tau}{2}sK_1}\,\dot
\Omega(r)\,\di r\right\}\\ &&= \ -\frac{2i}{\tau}\,
e^{-i\frac{\tau}{2}rK_1}\,X
\,e^{i\frac{\tau}{2}rK_1}\,\Omega(r)\Big|_0^s\,-\,\int_0^s
e^{-i\frac{\tau}{2}sK_1}\,X \,e^{i\frac{\tau}{2}sK_1}\,K_2(r)
\Omega(r)\,\di r \,.
\end{eqnarray*}
The first term is bounded in norm by $4\frac{\|X\|}{\tau}$, while
the second one is bounded by $\|X\|\cdot\|K_2\|$. It follows from
\eqref{eq:X} that $\|X\|\le \frac{2\|K_2\|}{ g_F}$. On the other
hand, $\|K_2\|=\delta_3 \|V\|=|E_F|\delta_3\le \delta_3$ by
\eqref{eq:HIest}, and the result follows.
\end{proof}
\section{Proof of Theorem \ref{thm:delta}}
The algorithm used in the proof is inspired by the mean ergodic theorem and makes use of the fact that the survival probability
$c_F (t) = \langle\psi_I|e^{itH_F}| \psi_I\rangle$ is directly measurable in AQC framework. We suggest to measure the survival probability for a number a times specified below to estimate the overlap $\delta_2$, and then to count the total running time spent on this subroutine.

Our starting point is a truncated Taylor's expansion for $e^x$:
\[ e^x\ =\ \sum_{k=0}^L \frac{x^k}{k!}\ +\ O\bigg(\frac{|x|^L}{L!}\bigg). \]
Setting $x=pe^{iw}$, and multiplying both sides by $x^{-p}$, we obtain the following relation:
\begin{equation}\label{eq:taylor}
e^{p(\cos w -1 )}e^{ip\sin w}\ =\ e^{-p}\sum_{k=0}^L \frac{p^ke^{iwk}}{k!}\
+\ O\bigg(e^{-p}\frac{p^L}{L!}\bigg).
\end{equation}
If $1-\cos w >g$, then the left hand side of Eq.~\eqref{eq:taylor} is bounded by $e^{-pg}$
and therefore is smaller than $1/N^2$, provided $p=2\ln N/\min(1, g)$. On the other hand, with such choice of $p$, the remainder term in Eq.~\eqref{eq:taylor} is bounded by
$O(1/N^2)$ if $L$ is chosen to be equal to $ep$. Combining these observations, we get
\begin{equation}\label{eq:N2}
e^{-p}\sum_{k=0}^{ep} \frac{p^ke^{ik\omega}}{k!}\,=\,
\begin{cases}
&1+O(1/N^2)\,, \hspace{1cm}\mbox{ if } \omega=0 \\
 &   \ \ O(1/N^2)\,, \hspace{1.4cm}\mbox{ if } 1-\cos\omega>g
\end{cases}\ \ \, ,
\end{equation}
where $p=2\ln N/\min(g,1)$.

Now, using the spectral decomposition of $H_F$,
\[e^{it(H_F-E_F)}=\sum_{i=1}^N P_ie^{it(E^f_i-E_F)}P_i,\]
where $E_i$ is the $i$-th distinct eigenvalue  of $H_F$ and $P_i$ is the projector onto the spectral subspace associated with $E_i$. Hence
\begin{equation}\label{eq:iden'}
e^{-p}\sum_{t=0}^{ep} \frac{p^k}{k!}\,\langle\psi_I|
e^{it(H_F-E_F)}\psi_I\rangle\,=\,(\delta_2)^2+O(1/N^2)\,,
\end{equation}
for $p=2\ln N/\min(1, 1-\cos g_F)$ where we have used Eq.~\eqref{eq:N2}.
The total running time is $\sum_{t=1}^{ep} t=O((\ln N)^2)$.
\hfill\ep

\section{Proof of Theorem \ref{thm:gap}}
The main tool we are going to use is the so called Krein's formula \cite{AG} for the rank $m$ perturbation of the initial Hamiltonian $H_I$. It
gives a characterization of the location of $m$ eigenvalues of the
perturbed matrix that {\it differ} from the spectral values of
$H_I$. Specifically, let $A,B$ be two hermitian matrices, with
$Rank\ B=m$, and let $Q$ be an orthogonal projection onto $Range\
B$. Suppose that $A$ is invertible (that is $0\notin\sigma(A)$). Then Krein's formula tells us that
\be\label{eq:krein}(A+tB)^{-1} \ = \ \left(K^{-1}+tBQ\right)^{-1}\,,\ee
with
\[ K\ :=\ QA^{-1}Q\,,\]
whenever the right hand side of \eqref{eq:krein} exists and where $K^{-1}+tBQ$  is interpreted as acting in the $m$--dimensional space
$Range\ B$.  In other words,  $0$ is an eigenvalue of $A+tB$ if and only if the $m\times m$ matrix $K^{-1}+tBQ$ contains $0$ in its spectrum.

The Krein's formula follows directly from the Schur
complement formula, which says that if $C$ is invertible then
\[QC^{-1}Q \ = \ \left(QCQ\,-\,QC\bar Q\,(\bar QC\bar Q)^{-1}\,\bar Q
CQ\right)^{-1}\,,\]
where $\bar Q:=I-Q$ and the inverses on the right hand side are
understood as acting on the ranges of $\bar Q$ and $Q$,
respectively.

To apply the Krein's formula in our context, we form a one parameter
family
\[H_t:=\ H_I+tH_F\,,\quad t=\frac{s}{1-s}\,,\quad
t\in[0,\infty)\,.\]
It is then follows that for a fixed $t\in(0,\infty)$ the eigenvalues
of $H_t$ that differ from $\sigma(H_I)$ are given by the roots of
the equation
\begin{equation}\label{eq:roots}
\det\left( K^{-1}(E)\,+\,tH_F Q_F\right)\ = \ 0\,,
\end{equation}
where
\[K(E):=Q_F\left(H_I\,-\,E\right)^{-1}Q_F\,.\]
Whenever it is clear from the context that we are working with the
operators on $Range\ Q_F$, we will suppress the $Q_F$ dependence.

To analyze \eqref{eq:roots}, we start with the following simple
observation:
\begin{lemma}\label{lem:posit}
The matrix $K(E)$ can be decomposed as
\begin{equation}\label{eq:decomp}
K(E)\ = \ \hat K(E) \,+\,\frac{\delta_4^2\,D}{E_I-E}\,.
\end{equation}
Here the matrix $D$ is positive semi-definite, and is bounded in
norm by $1$. The matrix $\hat K(E)$ is holomorphic in the half plane
$Re\ E>E_I-g_I/2$ and is positive definite for
$E\in[E_I-g_I/2,E_I+g_I/2]$. Moreover, in this interval we have
bounds
\begin{equation}\label{eq:Kbounds}
\frac{2}{4+g_I}-\delta_4^2\ \le \ \hat K(E)\ \le \
\frac{2}{g_I}\,;\quad \frac{4}{(4+g_I)^2}-\delta_4^2\ \le \
\frac{d\hat K(E)}{dE}\ \le \ \frac{4}{g_I^2}\,.
\end{equation}
\end{lemma}
\begin{proof}
We decompose
\begin{eqnarray*}K(E) & = & Q_F\left(H_I\,-\,E\right)^{-1}Q_F \\ &=&
\ Q_F\bar P_I\left(H_I\,-\,E\right)^{-1}Q_F\,+\,Q_F
P_I\left(H_I\,-\,E\right)^{-1}Q_F\,.
\end{eqnarray*}
The first contribution will correspond to $\hat K(E)$ in
\eqref{eq:decomp}, and the second one to its counterpart in
\eqref{eq:decomp}. Note now that for $E\in[E_I-g_I/2,E_I+g_I/2]$ we
have
\[\frac{g_I}{2}\,\bar P_I \ \le \ \bar P_I(H_I-E) \ \le\
\left(2\,+\,\frac{g_I}{2}\right)\,\bar P_I\,,\]
where the the upper bound is a consequence of $\|H_I\|=1$. Hence we
obtain
\[\frac{2}{4+g_I}\,Q_F\bar P_IQ_F\ \le \ Q_F\bar
P_I(H_I-E)^{-1}Q_F \ \le \ \frac{2}{g_I}\,Q_F\bar P_IQ_F\,.\]
Therefore, the first bound in \eqref{eq:Kbounds} follow now from
\[Q_F\bar P_IQ_F \ = \ Q_F\,-\,Q_F P_IQ_F\]
and
\[ 0\ \le \ Q_F P_IQ_F\ \le \ \delta_4^2\, Q_F\,.\]
To obtain the second bound in \eqref{eq:Kbounds} we note that
\[\frac{d}{dE}\left(H_I\,-\,E\right)^{-1} \ = \
\left(H_I\,-\,E\right)^{-2}\]
for $E\notin\sigma(H_I)$, and then proceed as above.
\end{proof}
In applications to the AQC the parameter $\delta_4$ is typically
extremely small: $\delta_4^2 = O(1/N)$. Hence the second contribution
in \eqref{eq:decomp} is small provided $|E-E_I|\gg \delta_4$.
Therefore for value of $E$ in such intervals, we can first find the
roots $\hat E_i(t)$ of
\begin{equation}\label{eq:rootsmod}
\det\left( \hat K^{-1}(E)\,+\,tH_F\right)\ = \ 0\,,
\end{equation}
and then estimate $|\hat E_i(t)- E_i(t)|$, where $E_i(t)$ are
corresponding roots of \eqref{eq:roots}. As we will see, the level
crossings or the avoided level crossings for $H_t$ occur for values
of $t$ such that a pair of eigenvalues $E_k(t), \ E_l(t)$ is close
to $E_I$. To find these values of $t$ in the first approximation, we
fix the value $E=E_I$ in \eqref{eq:rootsmod} and solve it for $t$.
We have
\begin{lemma}
The equation
\begin{equation}\label{eq:rootsmod'}
\det\left( \hat K^{-1}(E_I)\,+\,tH_F\right)\ = \ 0\,,
\end{equation}
has exactly $m_+$ roots $\{t_j\}_{j=1}^{m_+}$ on $(0,\infty)$, where
$m_+$ is a number of negative eigenvalues of $H_F$.
\end{lemma}
\begin{proof}
Let $A:=\hat K(E_I)$, then it follows from previous lemma that $0<
 A$. Hence
\[\left(A^{-1}\,+\,tH_F\right) \ = \
tA^{-1/2}\left(t^{-1}\,+\,A^{1/2}H_FA^{1/2}\right)A^{-1/2}\,.\]
The right hand side is not invertible for values $\{t_j\}$ such that
\[
-t_j^{-1}\,\in\,\sigma\left(A^{1/2}H_FA^{1/2}\right)\,,
\]
and the result follows now from  Sylvester's law of inertia
\cite{HJ}.
\end{proof}
We are now in position to estimate the size of the gap $g$ from
above. Namely, we consider the gaps $g_j$ for $H_t$ for $t=t_j$.
Since $t=\frac{s}{1-s}$ and $H(s)=(1-s)H_t$, we obtain $g\le
\frac{g_j}{1+t_j}\le g_j$. Let
\[\beta:= \
2\left(\frac{\delta_4^2}{\frac{4}{(4+g_I)^2}-\delta_4^2}\right)^{1/2}\,.\]
To get a bound on $g_j$ we show that \eqref{eq:roots} has roots in
the intervals $[E_I-\beta,E_I)$ and $(E_I,E_I+\beta]$, at $t=t_j$.
We then infer that $g_j\le 2\beta$, from which the upper bound in
\eqref{eq:gap} follows since $\beta<5\delta_4$. Observe first that by
condition of the Theorem \ref{thm:gap},
$\sigma(H_I)\cap[E_I-\beta,E_I)=\sigma(H_I)\cap(E_I,E_I+\beta]=\emptyset$ (where the latter property follows from the bound $\beta<5\delta_4<g_I/2$),  hence we are in position
to use Lemma \ref{lem:posit}. We only show that for the first
interval, the proof is analogous for the second one.

To this end, we will denote by $\sgn(A)$ the signature of the matrix
$A$. We observe that since $\frac{\delta_4^2\,D}{E_I-E}$ in
\eqref{eq:decomp} is positive semidefinite and monotone increasing
for the values of $E$ in $[E_I-\beta,E_I)$,  we have
\begin{equation}\label{eq:sgn1}
\sgn(K^{-1}(E_I-0)\,+\,t_jH_F) \ \le \ \sgn(\hat
K^{-1}(E_I)\,+\,t_jH_F)\,.
\end{equation}
On the other hand, we have
\begin{equation}\label{eq:rootsE}
\hat K(E_I) \, -\, \hat K(E_I-\beta)\ = \ \int_{E_I-\beta}^{E_I}\hat
K'(E)\,\di E \ \ge \ \beta
\left(\frac{4}{(4+g_I)^2}-\delta_4^2\right)\,,
\end{equation}
where in the last step we have used \eqref{eq:Kbounds}. Hence
\begin{eqnarray}K(E_I-\beta) & = & \hat
K(E_I-\beta)\,+\,\frac{\delta_4^2\,D}{\beta} \nonumber \\ & \le & \hat
K(E_I)\,-\, \left(\beta
\left(\frac{4}{(4+g_I)^2}-\delta_4^2\right)\,I\,-\,
\frac{\delta_4^2\,D}{\beta}\right)\\  &\le& \hat
K(E_I)\,-\,\frac{\delta_4^2}{\beta}\,I \,<\, \hat K(E_I)\nonumber \,,
\end{eqnarray}
with a choice of $\beta$ as above, and where we have used $\|D\|\le 1$. We
infer
\[K^{-1}(E_I-\beta)\,+\,t_jH_F\ > \ \hat
K^{-1}(E_I)\,+\,t_jH_F\,,\]
and since the matrix $\hat K^{-1}(E_I)\,+\,t_jH_F$ has zero
eigenvalue by construction, we obtain
\begin{equation}\label{eq:sgn2}
\sgn(K^{-1}(E_I-\beta)\,+\,t_jH_F) \ < \ \sgn(\hat
K^{-1}(E_I)\,+\,t_jH_F)\,.
\end{equation}
Combining \eqref{eq:sgn1} and \eqref{eq:sgn2} together, we get
\begin{equation}\label{eq:sgn3}
\sgn(K^{-1}(E_I-\beta)\,+\,t_jH_F) \ < \
\sgn(K^{-1}(E_I+0)\,+\,t_jH_F)\,.
\end{equation}
But the family $K^{-1}(E)\,+\,t_jH_F$ is continuous on
$[E_I-\beta,E_I)$, hence there should be some value of $E$ in this
interval for which $K^{-1}(E)\,+\,t_jH_F$ has the eigenvalue $0$.
\hfill\ep

\section{Proof of Theorem \ref{thm:at5}}
Let us remind the reader that in the context of this assertion $E_I=-1$. As in the proof of theorem \ref{thm:at1},
\[\phi_\tau(s):= \
e^{if_\tau(s)}\psi_\tau(s)\,,\quad f_\tau(s)\ =\ -\tau\,\int_0^s
(1-f(r))\,\di r\,.\]
and
\[\hat H(s)\ = \ (1-f(s))\,(H_I+1)\,+\,f(s)\,H_F\,.\]
\par
Let us introduce the auxiliary matrix
\[
B(s)=\left(f(s)H_F+1-f(s)+\epsilon i\right))^{-1}\,,
\]
and let $\phi(s) = \psi_I-f(s)H_FB(s)\, \psi_I$, where $\epsilon$ is
a small parameter to be chosen later. Omitting the $s$ dependence,
we have
\begin{equation}\label{eq:hatHp}
\hat H\phi= -f(1-f)H_IH_F B \,\psi_I -i\epsilon fH_FB \,\psi_I\,.
\end{equation}
That means that away from the $m$ values of $s$ for which $B^{-1}(s)$ has
zero eigenvalue, $\|\hat H\phi\|$ is very small, since
$\|H_IH_F\psi_I\|\le\delta^2$.  Note now that by fundamental theorem of calculus we have
\begin{equation}\label{eq:ftc}
\langle\phi(1)|\phi_\tau(1)\rangle \ = \
\langle\phi(0)|\phi_\tau(0)\rangle\,+\,\int_0^1\frac{d}{ds}
\langle\phi(s)|\phi_\tau(s)\rangle\,ds\,,
\end{equation}
where $\phi_\tau(s)$ is defined in \eqref{eq:phi'}. But
$\langle\phi(0)|\phi_\tau(0)\rangle=1$ and
\begin{eqnarray*}
\left|\langle\phi(1)|\phi_\tau(1)\rangle\right|& =&
\left|\langle\psi_I|\phi_\tau(1)\rangle \,-\,
\langle\psi_I|\frac{H_F}{H_F-\epsilon\,i}| \phi_\tau(1)\rangle\right|\\
&&\hspace{-2cm}\le\
\left|\langle\psi_I|\phi_\tau(1)\rangle\right|\,+\,\|Q_F\psi_I\|\
= \ \left|\langle\psi_I|\phi_\tau(1)\rangle\right|\,+\,\delta\,.
\end{eqnarray*}
Substitution into Eq.~\eqref{eq:ftc} gives
\[1\,-\,\left|\langle\psi_I|\phi_\tau(1)\rangle\right|\ \le \
\left|\int_0^1\frac{d}{ds}
\langle\phi(s)|\phi_\tau(s)\rangle\,ds\right| \,+\,\delta\,.\]
Hence Eq.~\eqref{eq:mintime''} will follow if
\begin{equation}\label{eq:nbn}
\left|\int_0^1\frac{d}{ds}
\langle\phi(s)|\phi_\tau(s)\rangle\,ds\right| \ <\
1\ -\ \frac{2\sqrt{6}}{5}\,-\,3\delta\,.
\end{equation}
To establish the above bound, we note first that
\be\label{eq:derph}\frac{d}{ds}
\langle\phi(s)|\phi_\tau(s)\rangle \ = \
\langle\dot\phi(s)|\phi_\tau(s)\rangle \,-\,i\tau\langle\phi(s)|\hat
H(s)|\phi_\tau(s)\rangle\,.\ee
We bound the first term on the right hand side by $\|\dot\phi\|$ and
the second one by $\tau\|\hat H\phi\|$.  On the other hand,
suppressing the $s$-dependence, we have
\[\dot \phi \ = \ -\dot f H_FB\psi_I\,-\,f\dot f B(H_F-1)B H_F\psi_I\,,\]
hence
\[\|\dot\phi\| \ \le\  \dot f\,\|B\| \,\|H_F\psi_I\|\,+\,2\dot f\,\| B\|^2\,\|H_F\psi_I\|\,,\]
where we have used $|f|\le 1$ and $\dot f\ge0$.  Let  $\dist(S,z)$ be an Euclidean distance from the set $S$  to the
point $z$ in $\C$, and let  $\sigma(H)$ stand for the spectrum of $H$. Then we can estimate the right hand side further as
\be\label{eq:derph'}\|\dot\phi\| \ \le\ \frac{\dot f\delta}{\Delta_\epsilon}\,+\,
\frac{2\dot f\delta}{\left(\Delta_\epsilon\right)^2}\,,\ee
with
\[\Delta_\epsilon(s):=\dist(f(s)\sigma(H_F)\,,\,-1+f(s)+\epsilon\,i)\,\]
and where we have used  $-Q_F\le H_F\le Q_F$.  Taking the norm from the both sides of \eqref{eq:hatHp} we get
that
\be\label{eq:derph'a} \|\hat H\phi\| \ \le\ \|H_I Q_F\|\,\|f H_F B\|\,\|Q_F\psi_I\|\,+\,\epsilon\,\|f H_F B\|\,\|Q_F\psi_I\| \,,
\ee
where we have used $H_F Q_F= H_F$ and  $\|H_I Q_F\|\le \|Q_I Q_F\|$, with the later relation following from
\[\|H_I Q_F\| \ = \ \|H_I Q_IQ_F\| \ \le \ \|H_I\|\,\|Q_IQ_F\|\,.\]
To estimate $\|f H_F B\| $ we consider three cases:
\begin{enumerate}
\item{$s\in[0,a]$:}  On this interval we can estimate
\be\label{eq:1/3}
\|f(s) H_F B(s)\| \ \le \ \max_{s\in[0,a]}\frac{1}{\Delta_\epsilon(s)} \ \le \ 3\,,\ee
where we have used $\sigma(H_F)\subset [-1,1]$.
\item{$s\in J$:} In this case, we  bound
\be\label{eq:J}\|f(s) H_F B(s)\| \ \le \ \frac{1}{\Delta_\epsilon(s)} \ \le \ \frac{\dot f(s)}{\kappa \Delta_\epsilon(s)}\ee
using theorem's hypothesis.
\item{$s\in[b,1]$:} Here we estimate
\be\label{eq:a}\|f(s) H_F B(s)\| \ = \ \left\|f(s) H_F \left(f(s)H_F+1-f(s)+\epsilon i\right))^{-1}\right\| \ \le \ 1\,+\,\frac{(1-f(s))+\epsilon}{\Delta_\epsilon(s)}\ \le \ 2\,+\,\frac{\dot f(s)}{\Delta_\epsilon(s)}\,,\ee
where in the last step we have used \eqref{eq:concave}.

\end{enumerate}
Plugging \eqref{eq:1/3} -- \eqref{eq:a} into \eqref{eq:derph'a}, we obtain
\be\label{eq:derph''} \|\hat H\phi\|  \ \le \
(\delta^2+\epsilon\delta)\left(3\,+\,\frac{\dot f(s)}{\kappa \Delta_\epsilon}\right)\,.
\ee
Using \eqref{eq:derph'} and  \eqref{eq:derph''} to bound the right hand side of  \eqref{eq:derph}, we get
\be\label{eq:ketbra}\left|\frac{d}{ds}
\langle\phi|\phi_\tau\rangle\right| \ \le \ \frac{\dot f\delta}{\Delta_\epsilon}\,+\,
\frac{2\dot f\delta}{\left(\Delta_\epsilon\right)^2}\,+\, \tau(\delta^2+\epsilon\delta)\left(3\,+\,\frac{\dot f(s)}{\kappa \Delta_\epsilon}\right)\,.\ee
In what follows we will use
\begin{lemma}\label{lem:suppl}
We have bounds
\be
\int_0^1\frac{\dot f ds}{\Delta_\epsilon}\ \le \
-2(m+1)\ln\epsilon\,;\quad  \int_0^1\frac{\dot f
ds}{\left(\Delta_\epsilon\right)^2}\ \le \
\frac{2(m+1)}{\epsilon}\,.\label{eq:intbnds}
\ee
\end{lemma}
Integrating both sides of \eqref{eq:derph} over $s$ and using \eqref{eq:ketbra} and \eqref{eq:intbnds}, we arrive at
\[
\left|\int_0^1\frac{d}{ds}
\langle\phi(s)|\phi_\tau(s)\rangle\,ds\right| \ \le\
3\tau(\delta^2+\epsilon\delta)\,+\,2(m+1)\delta\left(-\ln\epsilon\left(1+\frac{\tau\delta}{\kappa}+
\frac{\tau\epsilon}{\kappa}\right)+\frac{2}{\epsilon}\right)\,.
\]
Hence the required bound in Eq.~\eqref{eq:nbn} follows with the
choice $\epsilon=10^{3}(m+1)\delta$, provided $\tau\le
-\frac{C\kappa}{\epsilon^2\ln\epsilon}$ where $C$ is some generic constant.
 \hfill\ep
\begin{proof}[Proof of Lemma \ref{lem:suppl}]
 We derive the first bound, the second bound can be shown analogously. To this end, we observe that
\begin{eqnarray*}
\frac{\dot f }{\Delta_\epsilon} & = & \max_{E_n\in\sigma(H_F)}
\frac{\dot f }{|fE_n+1-f-\epsilon\,i|} \\ 
& < & \sum_{E_n\in\sigma(H_F)}\,\frac{\dot f }{|f(E_n-1)+1-\epsilon\,i|}\\ 
\end{eqnarray*}
 It follows that
\begin{eqnarray*}
\int^1_0 \frac{\dot f ds}{\Delta_\epsilon} & < & \sum_{E_n\in\sigma(H_F)}
\int^1_0 \frac{\dot f ds}{|fE_n+1-f-\epsilon\,i|} \\
& \le & (m+1)\max_{E\in[-1,1]}\int^1_0 \frac{\dot f ds }{|fE+1-f-\epsilon\,i|}\,,
\end{eqnarray*}
where $m=Rank \ H_F$. But
\begin{eqnarray*}&&\int_0^1\frac{d f }{\sqrt{\left(f(E-1)+1\right)^2+\epsilon^2}} \\ && \hspace{1cm}= \ \left.\frac{1}{E-1}\,\ln\left((E-1)f+1+\sqrt{((E-1)f+1)^2+\epsilon^2}\right)\right|_0^1\\ && \hspace{1cm}= \
\frac{1}{E-1}\,\ln\left(\frac{E+\sqrt{E^2+\epsilon^2}}{1+\sqrt{1+\epsilon^2}}\right)\,. \end{eqnarray*}
One can check by taking the derivative that the expression on the right hand side is monotonically decreasing for all $E\in[-1,1]$. Since it is also positive and continuous, this term achieves  its maximum  at $E=-1$, with the value
\[-\frac{1}{2}\ln\left(\frac{\sqrt{1+\epsilon^2}-1}{1+\sqrt{1+\epsilon^2}}\right) \ \le \ -2\,\ln\epsilon\]
for $\epsilon$ small enough.  Hence the first bound in \eqref{eq:intbnds} follows.
\end{proof}

\section{Conclusions}
In this work we derived a number of rigorous results concerning Hamiltonian-based quantum search problems that satisfy Assumption \ref{assump'}. Our results include in particular upper and lower bounds on the amount of time needed to perform a general Hamiltonian-based quantum search, a lower bound on the evolution time needed to perform a search that is valid in the presence of control error and a generic upper bound on the minimum eigenvalue gap for evolutions.

The lower bound on the evolution time is to our knowledge the tightest for
AQC type problems. It matches exactly the results established in the framework of  the quantum circuit algorithm, when  applied to Grover's search problem.

We then construct a specific Hamiltonian-based algorithm with the runtime $\tau=O(\sqrt{N/m})$ which nearly recovers the lower bound, and thus close to be optimal,  in the case of Grover's search where the final Hamiltonian is a  rank $m$  projection. We augment this construction with the Hamiltonian-based quantum counting subroutine,  which allows us to compute one input parameter $\delta_2$ to the algorithm. Since the algorithm is robust in a sense that we can allow the running time to vary in the large time interval, it is not very sensitive to the  ground state energy $E_F$. As a result, just an approximate knowledge of $E_F$ is needed to ensure that the algorithm works.

While our methods do not hinge on the knowledge of the gap structure of the underlying interpolating Hamiltonian $H(s)$, we  establish an upper bound on the size of the first spectral gap for this family of matrices.  For the final Hamiltonian of the Grover's type, {\it i.e.} $H_F$ is a  rank $m$  projection, the smallest value of the gap cannot exceed $O(\sqrt{m/N})$.

Lastly, we address the question of the the evolution time for search-Hamiltonians that are also error robust. Namely, we obtain the lower bound on the running time when the velocity $\dot f$ is greater than
a fixed control parameter $\kappa$ during the evolution, excluding the vicinities of the the endpoints $s=0,1$ where it is allowed to be small. We show that the necessary control accuracy requirement should be at least $O(\ln N/\sqrt{N})$ for the algorithm to succeed. In particular, this result implies that starting slow / finishing slow strategy by itself is not sufficient to make AQC better than the classical computer.

For a general form of $H_F$ for which our Assumption  \ref{assump'} is invalid no general lower bounds on the runtime can be obtained, as the examples constructed in  \cite{DMV} show. It will be very interesting to establish a "typical" lower and upper bounds for the random instances of NP-complete problems,  discussed in \cite{AKR}.

\begin{acknowledgments}  We are grateful to the referee for valuable comments and suggestions. This work was supported in part by NSF
grant DMS--0907165.
\end{acknowledgments}
%
%
\section*{References}
%
%

\end{document}